\documentclass{article}
\usepackage{graphicx} 
\usepackage{amsmath}
\usepackage{amsthm}
\usepackage{amssymb}

\newtheorem{proposition}{Proposition}[section]
\newtheorem{corollary}{Corollary}[proposition]
\newtheorem{lemma}{Lemma}[section]
\usepackage{nicematrix}
\usepackage[style=vancouver,sorting=none]{biblatex}
\usepackage{amsmath}
\usepackage{authblk}
\usepackage{hyperref}

\AtEveryBibitem{\clearfield{doi}\clearfield{month}\clearfield{day}}
\appto{\bibsetup}{\emergencystretch=4em\sloppy}

\newcommand{\fzero}{\mathbf{F}_{N,0,0,k}}
\newcommand{\fm}{\mathbf{F}_{N,m,k_1,k_2}}
\newcommand{\fn}{\mathbf{F}_{N,n,k_1,k_2}}

\title{Demographic inference of pathogen-infected populations from partially observed transmission forests}
\author[1]{Matthew Hall}
\affil[1]{Department of Clinical Research, London School of Hygiene and Tropical Medicine, London, UK}
\date{}

\begin{document}

\maketitle

\begin{abstract}
Genomic epidemiology has several methods for identifying which sampled individuals are linked by
close proximity in the transmission chain, whether by grouping them into clusters under a genetic
distance threshold or by identifying probable direct transmission pairs. Individuals found to have
no sampled neighbours---singletons---are usually set aside. We argue that they are informative.
Whether any two sampled individuals prove to be linked depends on the size of the sampling frame
and the number of independent lineage introductions into it, and thus the balance of linked and
unlinked individuals is itself data about these quantities.

We formalise this by treating the intersection of a transmission tree with a sampling frame as a
labelled rooted forest, from which a fixed number of nodes are sampled uniformly at random. Using
the all-minors matrix-tree theorem, we derive a closed-form expression for the number of labelled
rooted $k$-forests on $N$ nodes in which a specified set of nodes is independent. From this we
obtain, again in closed form, the probability that a sample of a given size contains no linked
individuals at all, and the likelihood of an arbitrary observed configuration of clusters, both when the
transmission structure within each is known, and when only the cluster sizes are. 

The approach is closer to a survey, or to mark-recapture, than to conventional phylodynamic model
fitting: the information comes from the linkage structure of a single sample alone with no
assumptions regarding pathogen dynamics. We outline three applications: power calculations for
prospective studies, a test of the uniform sampling assumption that underlies the reading of
clusters as transmission hotspots, and demographic inference itself. We finally set out the assumptions
that a more flexible implementation would need to relax.

\end{abstract}

\section{Introduction}

The fundamental principle of genomic epidemiology is that there is a correlation between similarity between pathogen genomes and the proximity of the hosts that they were derived from in the chain of transmission. It is this principle that allows a phylogeny to be used as a proxy for a transmission tree in the field of phylodynamics \cite{Grenfell2004-ph,Volz2013-vp}, and it is this principle that allows related cases to be grouped into ``clusters'' according to a numerical threshold for genomic similarity \cite{RagonnetCronin2013-cp,KosakovskyPond2018-ht}. More recent ``source attribution'' methodologies also allow for the identification of direct transmission pairs from genomic data \cite{Hall2015-bl, Colijn2024-br, Campbell2018-o2, Ypma2013-tt,Didelot2017-tp}, sometimes by leveraging within-host diversity \cite{Wymant2018-ps,DeMaio2018-bt, Skums2018-qn}.

Generally, clustering studies have been concerned solely with the investigation of individuals who are actually members of clusters \cite{Wertheim2014-gn,Oster2018-mc}, and source attribution studies the investigation of individuals who are linked to another individual by transmission \cite{Ratmann2019-ds}. This, of course, requires that other cluster members, or partners in transmission, have actually been sampled according to the protocol of the study. Both methodologies will usually also find ``singleton'' individuals whose neighbours have not been sampled, but these have rarely been the subject of investigation. Here, we propose that this is an omission and provide a basis for a methodology that would integrate the singletons in analysis. 

Suppose a ``study'' is to be conducted that involves the identification of individuals in a population who are, or were at some point, infected with a pathogen. Let the ``sampling frame'' represent the complete set of individuals who ever could be included in the study. We assume that if an individual is sampled, they were infected with the pathogen under investigation and, moreover, the act of sampling provides a pathogen genome. The individuals making up the sampling frame are nodes in the transmission tree that represents the history of transmissions of that pathogen, but the entire transmission tree need not be in the frame \cite{Ypma2013-tt,Didelot2017-tp} and in fact almost never will be. In the vast majority of cases, at least some individuals who have infected or been infected by members of the sampling frame will lie outside it (for example, because they live in a location where sampling was not performed, or were infected before the study was even conceived). Restricting nodes in the tree to individuals in the sampling frame breaks the tree into a forest. Each component of the forest is a transmission tree, rooted at an individual whose infector was not in the frame. See figure~\ref{fig:examplett} for an example.

There is a difference between being a member of the sampling frame, and thus eligible for sampling, and actually being sampled. Suppose that a fixed number, $n$, of the nodes in the forest are actually sampled, and they are sampled uniformly at random. There are two properties of the forest that affect how many of the $n$ are singletons. The first is the number of components in the forest. The more of these there are, the greater the chance that two sampled individuals belong to different components. The second is just the overall number of nodes. A forest being larger will decrease the chance that two randomly-selected individuals are adjacent to each other, and thus increase the chance that any given individual is a singleton.

The distribution of singletons and non-singletons is therefore pertinent to the overall size of the sampling frame, and the number of separate introductions of lineages to the population it represents. In this paper, we explore the mathematics underlying this observation, deriving, using elementary tree combinatorics, the likelihood of observing an arrangement of sampled cases for a forest of a given size. There are three key applications of the methodology:

\begin{itemize}
    \item \textbf{Power calculations:} under assumptions about the forest, we can calculate how many samples we need to take in order to obtain a dataset with a given number of linked individuals, or, to put it another way, a desired number of transmission pairs.
    \item \textbf{Hypothesis testing:} the results here are obtained under the assumption that individuals are sampled uniformly at random. This is also a fairly common assumption in clustering studies \cite{Novitsky2014-sd,Frost2015-dc,Poon2016-gc}. Sampling may be assumed not to be clustered and the presence of an individual in the sample not to make it more likely that their direct contacts will also be in the sample. This is contrary to any ``word of mouth'' effect where individuals obtaining a diagnosis may tell their personal contacts to take a test. The framework here allows us to test real datasets for departure from that assumption.
    \item\textbf{Inference:} If the sampling assumptions are reasonable, we can estimate the total size of the sampling frame, and/or the number of lineage introductions to it \cite{duPlessis2021-li}. Neither of these is necessarily straightforward to obtain using other means \cite{Hook1995-cr}. 
    \end{itemize}

\begin{figure}
    \centering
    \includegraphics[width=0.9\linewidth]{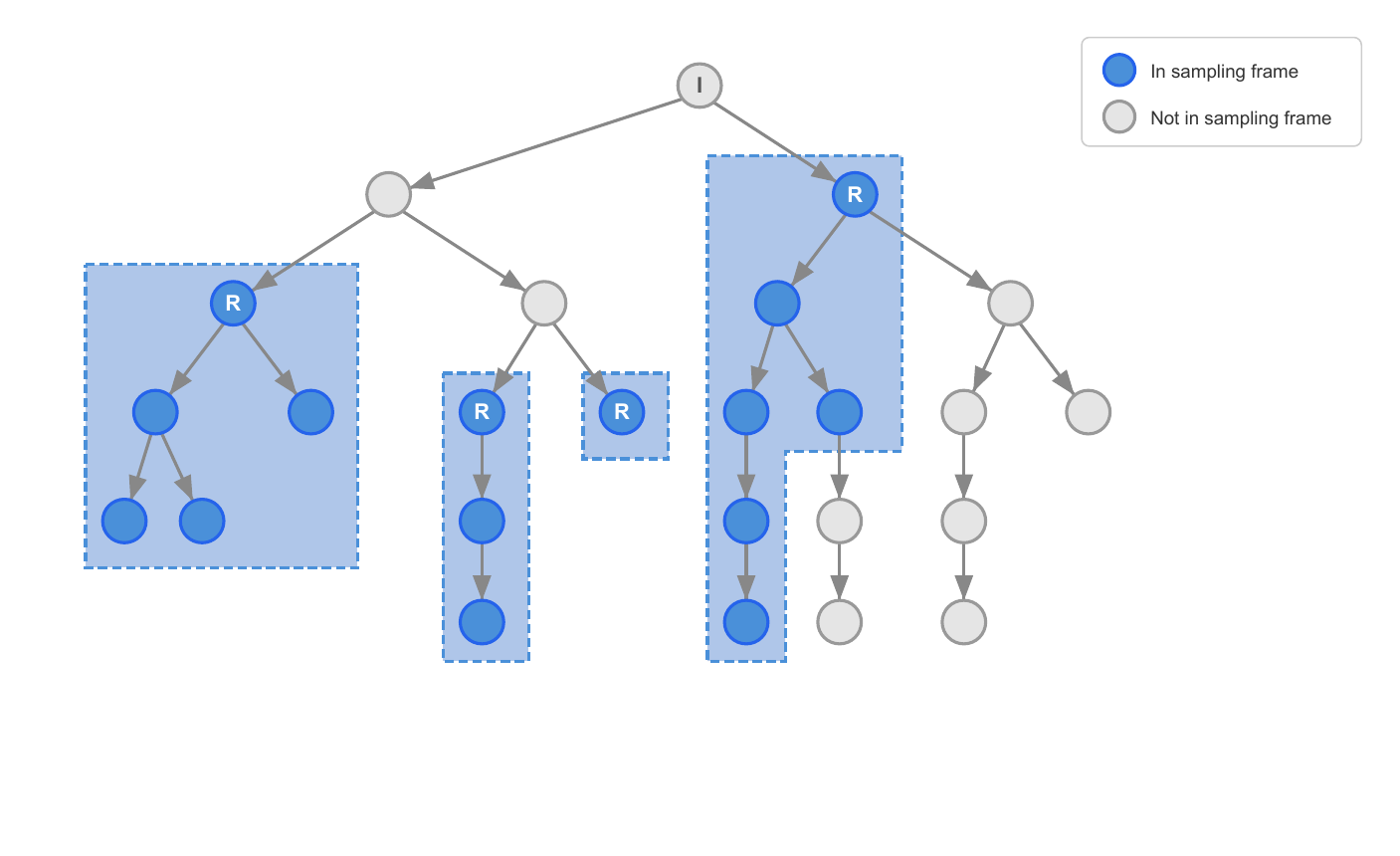}
    \caption{A full transmission tree descended from an index infection (I) but in which only 15 individuals are in the sampling frame (blue). The intersection of the sampling frame with the tree produces a forest with four components, each of which has its own root node (R). One sampled individual is a singleton.}
    \label{fig:examplett}
\end{figure}

\section{Labelled $k$-forests with specified independent nodes}\label{section1}

Let $\fzero$ be the complete set of labelled, rooted $k$-forests with $N$ nodes, where the roots are known; the labels run from $1$ to $N$. Nodes will be named only by their labels. The size of $\fzero$ is $kN^{N-k-1}$ \cite{Moon1970-li} (for $k=1$ this is Cayley's formula). If $k_1,k_2\in\mathbb{Z}^{\geq 0}$ are such that $k=k_1+k_2$, let $\fn$ be the set of labelled, rooted $k$-forests such that $n$ specific nodes are independent (not adjacent to each other), and $k_1$ of the $k$ roots are amongst these $n$ while the remaining $k_2$ are not. (Note $\fm\subset\fn$ for $m>n$, and $|\fzero|=|\mathbf{F}_{N,1,0,k}|=|\mathbf{F}_{N,1,1,k-1}|$.) 


\begin{proposition}\label{p1}
$|\fn|=(kN - k_1(n+k_2))N^{N-n-k_2-1}(N-n)^{n-k_1-1}$.
\end{proposition}
(Notice that this is equal to $kN^{N-k-1}$ if $n=0$ and $k_2=k$, or if $n=1$ and either $k_2=k$ or $k_2=k-1$.) 

\begin{proof}
From the complete $N$-graph, remove all the edges connecting the first $n$ nodes, to make a graph $C^n$. The Laplacian matrix $L$ of $C^n$, with rows and columns ordered as our labels, consists of a top-left $n\times n$ block which has all diagonal entries equal to $N-n$ and $0$ elsewhere, a bottom-right $(N-n)\times (N-n)$ block which has all diagonal entries equal to $N-1$ and $-1$ elsewhere, and all entries outside these blocks equal to $-1$:

\begin{align*}
L = \begin{bNiceArray}{cccc|cccc}[margin]
N-n & 0 & \cdots & 0 & \Block{4-4}<\Large>{-1}&&& \\
0 & N-n & & 0 &&&&\\
\vdots & & \ddots & \vdots &&&&\\
0 & 0 & \cdots & N-n &&&& \\
\hline
\Block{4-4}<\Large>{-1}&&&& N-1 & -1 & \cdots & -1 \\ 
&&&& -1 & N-1 &  & -1 \\
&&&& \vdots & & \ddots & \vdots \\
&&&& -1 & -1 & \cdots & N-1 \\
\CodeAfter\OverBrace[shorten,yshift=3pt]{1-1}{8-4}{n}
\CodeAfter\OverBrace[shorten,yshift=3pt]{1-5}{8-8}{N-n}
\end{bNiceArray}
\end{align*}
That the set of spanning unrooted $k$-forests of $C^n$ is the set of unrooted $k$-forests where the first $n$ nodes are independent should be clear after a little thought (any edges except those connecting the first $n$ nodes are permitted in members of the latter set). If $k=1$, the matrix-tree theorem \cite{Moon1970-li} states that the number of unrooted spanning trees of $C^n$ is the determinant of the reduced Laplacian, the matrix minor when removing one row and column from $L$. 
For our purposes we want the set of spanning \emph{rooted} $k$-forests of $C^n$ where the roots are specified. For this we can use the All-Minors Matrix-Tree Theorem \cite{Chaiken1982-pq} which dictates that $|\fn|$ is the minor determinant of order $N-k$ of $C^n$'s Laplacian, where the removed rows and columns are those from $1$ to $k_1$ and from $n+1$ to $n+k_2$. Call this matrix $L^{\mathrm{red}}$:

\begin{align*}
L^\mathrm{red} = \begin{bNiceArray}{cccc|cccc}[margin]
N-n & 0 & \cdots & 0 & \Block{4-4}<\Large>{-1}&&& \\
0 & N-n & & 0 &&&&\\
\vdots & & \ddots & \vdots &&&&\\
0 & 0 & \cdots & N-n &&&& \\
\hline
\Block{4-4}<\Large>{-1}&&&& N-1 & -1 & \cdots & -1 \\ 
&&&& -1 & N-1 &  & -1 \\
&&&& \vdots & & \ddots & \vdots \\
&&&& -1 & -1 & \cdots & N-1 \\
\CodeAfter\OverBrace[shorten,yshift=3pt]{1-1}{8-4}{n-k_1}
\CodeAfter\OverBrace[shorten,yshift=3pt]{1-5}{8-8}{N-n-k_2}
\end{bNiceArray}
\end{align*}
Bearing in mind the standard rules regarding elementary row operations preserving determinants, taking the first row and adding all the other rows to it gives:
\\
\begin{align*}
\mathrm{det}(L^\mathrm{red}) = 
\begin{vNiceArray}{cccc|cccc}[margin]
k_2 & k_2 & \cdots & k_2 & k & k & \cdots &  k\\
\hline
0 & N-n & & 0 & \Block{3-4}<\Large>{-1} &&&\\
\vdots & & \ddots & \vdots &&&&\\
0 & 0 & \cdots & N-n &&&& \\
\hline
\Block{4-4}<\Large>{-1}&&&& N-1 & -1 & \cdots & -1 \\ 
&&&& -1 & N-1 &  & -1 \\
&&&& \vdots & & \ddots & \vdots \\
&&&& -1 & -1 & \cdots & N-1 \\
\CodeAfter\OverBrace[shorten,yshift=3pt]{1-1}{8-4}{n-k_1}
\CodeAfter\OverBrace[shorten,yshift=3pt]{1-5}{8-8}{N-n-k_2}
\end{vNiceArray}
\end{align*}
and then by adding the first row multiplied by $\frac{1}{k_2}$ to each of the last $N-n-k_2$ rows:
\\
\begin{align*}
\mathrm{det}(L^\mathrm{red})=\begin{vNiceArray}{cccc|cccc}[margin]
k_2 & k_2 & \cdots & k_2 & k & k & \cdots &  k\\
\hline
0 & N-n & & 0 & \Block{3-4}<\Large>{-1} &&&\\
\vdots & & \ddots & \vdots &&&&\\
0 & 0 & \cdots & N-n &&&& \\
\hline
\Block{4-4}<\Large>{0}&&&& N-1+\frac{k}{k_2} & -1+\frac{k}{k_2} & \cdots & -1+\frac{k}{k_2} \\ 
&&&& -1+\frac{k}{k_2} & N-1+\frac{k}{k_2} &  & -1+\frac{k}{k_2} \\
&&&& \vdots & & \ddots & \vdots \\
&&&& -1+\frac{k}{k_2} & -1+\frac{k}{k_2} & \cdots & N-1+\frac{k}{k_2} \\
\CodeAfter\OverBrace[shorten,yshift=3pt]{1-1}{8-4}{n-k_1}
\CodeAfter\OverBrace[shorten,yshift=3pt]{1-5}{8-8}{N-n-k_2}
\end{vNiceArray}
\end{align*}
With the determinant of the bottom right block being $\frac{1}{k_2}(kN-k_1(n+k_2))N^{N-n-k_2-1}$, this is $(kN - k_1(n+k_2))N^{N-n-k_2-1}(N-n)^{n-k_1-1}$ as needed.
\end{proof}

\section{The bridge to epidemiology}

Returning to the scenario outlined in the introduction, suppose the sampling frame has $N$ individuals infected from $k$ independent lineage introductions, and $s$ are actually sampled, we assume uniformly at random. We also assume that, if we have sampled two individuals, we can tell perfectly whether they make up a direct transmission pair or not (although not necessarily the direction of transmission). 

\begin{proposition}\label{proppp}
    If the transmission forest $\mathcal{F}$ of the sampling frame has $k$ components and $N$ nodes, then the probability of identifying no pairs in $s$ draws is
        \begin{align*}
        \frac{1}{k{N\choose k}}\sum_{k_1=0}^k {s\choose k_1}{{N-s}\choose k-k_1}(kN - k_1(s+(k-k_1)))N^{k_1-s}(N-s)^{s-k_1-1}
    \end{align*}
\end{proposition}
\begin{proof}
    We can assume without loss of generality that the nodes we have sampled have the labels 1 to $s$. If we find no pairs, then these $s$ nodes in $\mathcal{F}$ are independent.

    The total number of rooted $k$-forests with $N$ nodes and specified roots is $kN^{N-k-1}$ \cite{Moon1970-li}. If the roots are unknown, then there are ${N\choose k}$ ways of picking them, giving a total of ${N\choose k}kN^{N-k-1}$ forests.

    The total number of rooted $k$-forests with $N$ nodes and the first $s$ independent, where $k_1$ roots lie between $1$ and $s$ and $k_2$ do not, is $(kN - k_1(s+k_2))N^{N-s-k_2-1}(N-s)^{s-k_1-1}$ by proposition~\ref{p1}. For a given $k_1$ there are ${s\choose k_1}{{N-s}\choose k_2}$ ways of picking the exact roots.

    The total number of forests where the $s$ are independent is thus:
    \begin{align*}
        \sum_{k_1=0}^k {s\choose k_1}{{N-s}\choose k-k_1}(kN - k_1(s+(k-k_1)))N^{N-s-(k-k_1)-1}(N-s)^{s-k_1-1}
    \end{align*}
    giving the result when divided by ${N\choose k}kN^{N-k-1}$.
\end{proof}

We must now consider situations where not every sampled individual is unlinked to all others. If we have sampled $s$ individuals, suppose the results of the genomic analysis are that they form $n$ connected, rooted subtrees $\mathcal{C}_1,\ldots,\mathcal{C}_n$ where $\mathcal{C}_i$ has $s_i>0$ nodes. We label the individuals such that the nodes in $\mathcal{C}_1$ have labels $1$ to $s_1$, those in $\mathcal{C}_2$ have labels $s_1+1$ to $s_1+s_2$, and so on. We assume we know the $s_i$s perfectly from the data. For the $\mathcal{C}_i$s, there are two possible situations:
\begin{enumerate}
    \item We also know each $\mathcal{C}_i$; we know the exact ancestral relationship between the individuals in the sample where those exist, including which was the first infected (the root of $\mathcal{C}_i$).
    \item We do not know the $\mathcal{C}_i$s. We know that the members of a subtree form a cluster; all are related to one of the others by direct transmission, but we do not know exactly who infected whom.
\end{enumerate}
In either case, let $S=(s_1,\ldots,s_n)$, so that $\sum_{i=1}^{n}s_i=s$. Let $\mathbf{C}$ be the ordered sequence of subtrees. We would like to be able to calculate the likelihoods:
\begin{align}\label{tractable}
    L(k,N|S) = p(S|k,N)
\end{align}
or
\begin{align}\label{tractable2}
    L(k,N|S, \mathbf{C}) = p(S,\mathbf{C}|k,N)
\end{align} 
We assume we know nothing about the true forest $\mathcal{F}$ prior to collecting the data other than that it is a $k$-forest with $N$ nodes, and as such the number of possible $\mathcal{F}$s is the number of forests of this type, i.e. ${N\choose k}kN^{N-k-1}$. The data, however, constrains that set, because there is a set of subtrees that must exist. Let $\mathbf{F}_S$ be the subset of elements $\mathcal{F}\in\fzero$ such that:
\begin{itemize}
\item{For all $i$, every numbered node in $\mathcal{C}_i$ is adjacent in $\mathcal{F}$ to another node in $\mathcal{C}_i$ (if there is one).}
\item{For all $i$ and $j$ with $i\neq j$, no numbered node in $\mathcal{C}_i$ is adjacent in $\mathcal{F}$ to a node in $\mathcal{C}_j$.}
\end{itemize}
and let $\mathbf{F}_\mathbf{C}$ be such that:
\begin{itemize}
\item{$\mathcal{F}$ contains every element of $\mathbf{C}$ as a subtree, and the path joining every non-root node of $\mathcal{C}\in\mathbf{C}$ to the root of the component of $\mathcal{F}$ that $\mathcal{C}$ is a subtree of runs through the root of $\mathcal{C}$. (In other words, the rooting of each $\mathcal{C}$ respects the rooting of $\mathcal{F}$).}
\item{For all $i$ and $j$ with $i\neq j$, no numbered node in $\mathcal{C}_i$ is adjacent in $\mathcal{F}$ to a node in $\mathcal{C}_j$.}
\end{itemize}
To calculate the likelihoods~\eqref{tractable} and \eqref{tractable2}, we need to enumerate $\mathbf{F}_S$ and $\mathbf{F}_\mathbf{C}$. Let us start with $\mathbf{F}_\mathbf{C}$.

\begin{proposition}\label{propFC}
For $0\leq k_1 \leq k$, let $k_2 = k-k_1$. Then:
\begin{align*}
|\mathbf{F}_\mathbf{C}| =& \sum_{k_1=\max\bigl(0,\,k-(N-s)\bigr)}^{\min(n,k)}\binom{N-s}{k_2}(N-s)^{n-k_1-1}N^{N-s-k_2-1}\\
&\qquad\times\sum_{\substack{R\subseteq[n]\\|R|=k_1}}\Bigl(k_2N+(N-s-k_2)\sum_{j\in R}s_j\Bigr)
\end{align*}
\end{proposition}
\begin{proof}
For a rooted forest $\mathcal{F}\in\mathbf{F}_\mathbf{C}$, let $\mathcal{G}_\mathbf{C}$ be the rooted $k$-forest obtained by collapsing every $\mathcal{C}_i\in\mathbf{C}$ to a single node. (In some cases a $\mathcal{C}_i$ will be a single node and nothing will happen to it, but for convenience we say this has still been ``collapsed''.) Renumber the nodes of $\mathcal{G}_\mathbf{C}$ such that the $n$ nodes that are the results of collapses come first. Let $d_i$ be the degree of the node acquired by collapsing $\mathcal{C}_i$. $\mathcal{G}_\mathbf{C}$ has $N - (s-n)$ nodes, $n$ of which must be independent. When $\mathcal{G}_\mathbf{C}$ is re-expanded to a member $\mathcal{F}\in\mathbf{F}_\mathbf{C}$, each of the $d_i$ incident edges to node $i$ has $s_i$ possible attachment points.

From the complete $(N-(s-n))$-graph, remove all edges connecting the first $n$ nodes, and then for all $i\leq n$ replace the edges joining all the nodes in $\{n+1, \ldots, N-(s-n)\}$ to node $i$ with $s_i$ separate edges. Let this graph be $D^{n,S}$. Every spanning $k$-forest of $D^{n,S}$ corresponds to an element of $\mathbf{F}_\mathbf{C}$, as if an edge connecting any node $j$ with a node $i\leq n$ exists, there are $s_i$ choices for it. We can, as in proposition~\ref{p1}, use the All-Minors Matrix-Tree Theorem on $D^{n,S}$. 

\begin{lemma}\label{matrixlemma}
Let $a,b\in\mathbb{Z}^+$ and $l,m\in\mathbb{Z}^{\geq0}$ with $l<a$ and $m<b$. Let $\{s_i\}$ be a sequence of $a$ positive integers. Let $L$ be a matrix of the form:

\begin{align*}
{\scriptsize
L= \begin{bNiceArray}{cccc|cccc}[margin]
s_1b & 0 & \cdots & 0 & -s_1 & -s_1 &\cdots & -s_1 \\
0 & s_2b & & 0 & -s_2 & -s_2 &\cdots & -s_2\\
\vdots & & \ddots & \vdots & \vdots & \vdots & & \vdots\\
0 & 0 & \cdots & s_{a-l}b & -s_{a-l} & -s_{a-l} &\cdots & -s_{a-l} \\
\hline
-s_1 & -s_2 & \cdots & - s_{a-l} & \sum_{i=1}^{a}s_i + b-1  & -1 & \cdots & -1 \\ 
-s_1 & -s_2 & \cdots & - s_{a-l} & -1 & \sum_{i=1}^{a}s_i + b-1 &  \cdots & -1 \\
\vdots & \vdots  &  & \vdots & \vdots &\vdots & \ddots & \vdots \\
-s_1 & -s_2 & \cdots & -s_{a-l} & -1 & -1 & \cdots & \sum_{i=1}^{a}s_i + b-1 \\
\CodeAfter\OverBrace[shorten,yshift=3pt]{1-1}{8-4}{a-l}
\CodeAfter\OverBrace[shorten,yshift=3pt]{1-5}{8-8}{b-m}
\end{bNiceArray}}
\end{align*}

Then
\begin{align*}
\mathrm{det}(L) = \left(\prod_{i=1}^{a-l} s_i\right) b^{a-l-1} \Bigl(\textstyle\sum_{i=1}^{a}s_i+b\Bigr)^{b-m-1}\left(m\sum_{i=1}^{a}s_i+ (b-m)\sum_{i=a-l+1}^{a}s_i + mb \right)
\end{align*}
\end{lemma}
\begin{proof}
This follows from elementary row operations in a similar way to proposition~\ref{p1}.
\end{proof}

The Laplacian matrix $L$ of $D^{n,S}$ is of the form of lemma~\ref{matrixlemma} with $a=n$, $b=N-s$ and $l=m=0$. First let us count the subset of $\mathbf{F}_\mathbf{C}$ with specified roots. The number of independent nodes that must be roots is $k-(N-s)$, or 0 if this is negative. The maximum number of independent nodes that can be roots is the smaller of $k$ and $n$. So let $\textrm{max}(0,k-(N-s))\leq k_1\leq\textrm{min}(n,k)$, and $k_2=k-k_1$. The $k$-forest count is the minor determinant of $L$ with the rows and columns corresponding to all the roots removed. For the non-independent roots this simply means removing any $k_2$ rows and columns from the last $b$, as those are identical, but for the independent roots let $R$ be the subset of $[n]$ of size $k_1$ containing them. We have a matrix of the form of lemma~\ref{matrixlemma} with $a=n$, $b=N-s$, $l=k_1$, $m=k_2$, and thus, writing $\sum_{i\in[n]}s_i=s$, the determinant is
\begin{align*}
\Biggl(\prod_{i\in[n]\setminus R}s_i\Biggr)(N-s)^{n-k_1-1}N^{N-s-k_2-1}\left(k_2N+ (N-s-k_2)\sum_{i\in R}s_i  \right).
\end{align*}
This counts each of the $d_i$ edges incident to collapsed node $i$ as having $s_i$ attachment points. For $i\notin R$ one of those edges joins $\mathcal{C}_i$ to the root of its component of $\mathcal{F}$, and the requirement that this path run through the root of $\mathcal{C}_i$ fixes its attachment point. The determinant therefore overcounts $\mathbf{F}_\mathbf{C}$ by a factor of $s_i$ for each $i\notin R$, which cancels the leading product exactly, leaving
\begin{align*}
C_R = (N-s)^{n-k_1-1} N^{N-s-k_2-1}\left(k_2N+ (N-s-k_2)\sum_{i\in R}s_i  \right).
\end{align*}

If we relax the assumption that the non-independent roots are known then this is multiplied by ${N-s}\choose{k_2}$ to count the possibilities; if we further relax and let the independent roots be unknown we need to sum over all subsets $R$ of $[n]$ that are of size $k_1$, and finally if we let $k_1$ be unknown then we need to sum over all values of it from $\textrm{max}(0,k-(N-s))\leq k_1\leq\textrm{min}(n,k)$ giving a total equal to:
\begin{align*}
\sum_{k_1=\max\bigl(0,\,k-(N-s)\bigr)}^{\min(n,k)}\;\sum_{\substack{R\subseteq[n]\\|R|=k_1}}\binom{N-s}{k_2} C_R
\end{align*}

\end{proof}


\begin{corollary}\label{corrrr}
\begin{align*}
|\mathbf{F}_S| = \prod_{i=1}^n s_i^{s_i-1} |\mathbf{F}_\mathbf{C}| 
\end{align*}
\end{corollary}
\begin{proof}
If only the numbers of nodes $s_i$ in the $i$th connected component of the data are known, then the number of rooted trees into which that component's nodes can be arranged is $s_i^{s_i-1}$ by Cayley's formula.
\end{proof}

Proposition~\ref{propFC} and Corollary~\ref{corrrr} give the total number of forests consistent with the data, in the former case where the $\mathcal{C}_i$ are known, and in the latter where only the $s_i$ are. To turn these into the likelihoods~\eqref{tractable} and \eqref{tractable2}, we simply need to divide by the number of possible $k$-forests on $N$ nodes, i.e. ${N\choose k}kN^{N-k-1}$.

\section{Discussion}

We propose here an approach to demographic inference using pathogen genomes that is markedly different to those that have come before. It has two established cousins in epidemiology, however. The first is mark-recapture, where a population is surveyed in multiple rounds and the number of individuals identified more than once is used to estimate the size of the population \cite{Hook1995-cr,  IWGDMF1995-cr, Chao2001-cr}. We sample only once and, instead, identify transmission links.  More recaptures correspond to a small population in mark-recapture, and more links a small population here. The key methodological difference we confront is that the transmission tree structure needs to be dealt with.

The second methodology is the snowball sampling for the estimation of network size \cite{Frank1994-hp,  Handcock2014-hp, Crawford2018-hp}. That too uses the links identified by epidemiological sampling to estimate the size of a population; the differences are, firstly, that the object of interest is a network not a tree, and, secondly, that individuals are sampled specifically because of their proximity to other individuals in the network. Here sampling is random.

There are two main types of source attribution methodology, which we will term ``phylodynamic'' and ``observational''. The former, which is more common \cite{Hall2015-bl, Colijn2024-br, Campbell2018-o2, Ypma2013-tt,Didelot2017-tp, Klinkenberg2017-pb}, treats a sample of genomes and their accompanying metadata as the result of the combination of an epidemiological transmission model, a mutation model, and a sampling model; the rules under which cases emerge and enter the dataset are parameterised. The latter instead collects genomes and establishes relationships between them based largely on those genomes alone. It is in some ways rather closer to a survey than a model-fitting exercise. The most prominent example of a methodology that works like that is phyloscanner \cite{Wymant2018-ps} and its derivatives \cite{Dhar2022-tn, Sledzieski2021-tf}. However, the entire realm of genomic clustering \cite{Croucher2015-bc,KosakovskyPond2018-ht,Stimson2019-st} is a relaxed version of the same idea. Clustering does not establish direction of transmission or even imply that the links it draws are direct transmissions at all, but it nevertheless finds close neighbours in the transmission tree based solely on genomes. It is dynamics-agnostic.

If the observational approach is to be likened to a survey, in which a sample of infected individuals are identified and then linked by transmission using their pathogen genomes, then the mathematical theory of the process is underdeveloped when compared to the phylodynamic approach. Previous statistical approaches have fit a mechanistic model of transmission between populations on top of the set of pairs \cite{Xi2022-pf, Bu2024-pp}. Here we take a different approach, leaning in to the similarities with a survey in which we are sampling random nodes from the transmission forest making up the sampling frame with the ability to identify which are transmission pairs, or form clusters. This will have strengths and weaknesses when compared to a more mechanistic approach, but fundamentally it is different, and adds a new tool to the box.

We propose three immediate use cases. Firstly, there is the question of how deeply one needs to sample the population of infected individuals in order to acquire a dataset with a desired number of transmission pairs. One of the key objectives in doing source attribution work is to characterise both the population of sources and the collection of identified pairs \cite{Ratmann2019-ds}; the collection of pairs might then be the subject of a separate statistical analysis in which each is a data point. Knowing the likely size of that dataset of pairs is essential to planning a prospective study.

The second use case is to check if the assumptions made regarding the sampling process and the epidemiology of the infection are valid. In HIV clustering it is quite frequent to treat identified clusters as hotspots of transmission \cite{Poon2016-hs,Wertheim2018-cg, Little2014-nw}, but it has been demonstrated that this may in some cases be artefactual, caused instead by clustered sampling \cite{Poon2016-gc}. There is a framework here for a sanity check, although the current version has caveats that we outline below.

Finally, there is the possibility of genuine demographic inference, whereby an epidemic size, or number of introductions to the sampling frame, is to be estimated from the arrangment of pairs and singletons in the dataset. This may be best in a Bayesian framework where priors can be placed on either quantity.

The methods outlined here are more a proof of concept than a fully usable suite of methods, and the assumptions made are quite stringent and indeed unrealistic. If the dataset consists of transmission pairs then it is assumed that those can be identified perfectly. If it consists of clusters then the assumption is that, while the transmission structure within a cluster is unknown, there are no missing intermediate individuals involved; everyone in the cluster except the root individual was infected by someone else in the cluster. This level of certainty in transmission reconstruction is never going to be realistic, although the exact sensitivity and specificity in declaring a pair of sequences to come from a transmission pair will vary greatly between organisms. A more robust iteration of this method would allow for error in making the inference. The method in its current iteration also disposes of all information from the genomes after pairs are called; proximity in the phylogeny is not used any further.

It is also assumed that all transmission histories are equally likely amongst the sampled cases. This flatly ignores an absolutely commonplace piece of information---sampling dates---but also ignores the particular characteristics that a pathogen brings to transmission tree shape, such as degree distribution. Both should be relaxed in refinements.

As large prospective genomic studies become more achievable and routine, the time is right for new approaches to analysing data of this kind. Proposed here are the foundations of a novel methodological approach, abandoning the ``dynamics-first'' principles of phylodynamics and instead treating genomic sampling as a process by which sections of an unknown transmission tree can be uncovered. It is a potentially useful addition to the toolbox, and further work is needed to explore its strengths and limitations.

\printbibliography

\end{document}